\documentclass[9pt,twocolumn]{IEEEtran}
\usepackage{amsmath ,amssymb,euscript ,yfonts,psfrag,latexsym,dsfont,graphicx,bbm,color,amstext,wasysym,subfig,parskip,soul}
\graphicspath{{./},{./figures/}}

\begin{document}
\newtheorem{thm}{Theorem}
\newtheorem{cor}[thm]{Corollary}
\newtheorem{conj}[thm]{Conjecture}
\newtheorem{lemma}[thm]{Lemma}
\newtheorem{prop}[thm]{Proposition}
\newtheorem{problem}[thm]{Problem}
\newtheorem{remark}[thm]{Remark}
\newtheorem{defn}[thm]{Definition}
\newtheorem{ex}[thm]{Example}

\newcommand{\mD}{{\mathbb D}}
\newcommand{\mE}{{\mathbb E}}
\newcommand{\mR}{{\mathbb R}}
\newcommand{\cN}{{\mathcal N}}
\newcommand{\cR}{{\mathcal R}}
\newcommand{\cU}{{\mathcal U}}
\newcommand{\HH}{{\mathrm H}}
\newcommand{\diag}{\operatorname{diag}}
\newcommand{\trace}{\operatorname{trace}}
\newcommand{\ignore}[1]{}
\newcommand{\mike}{\color{magenta}}
\newcommand{\rike}{\color{red}}
\newcommand{\bike}{\color{blue}}
\newcommand{\zero}{{0}}

\def\spacingset#1{\def\baselinestretch{#1}\small\normalsize}
\setlength{\parskip}{10pt}
\setlength{\parindent}{20pt}
\spacingset{1}

\definecolor{grey}{rgb}{0.6,0.3,0.3}
\definecolor{lgrey}{rgb}{0.9,.7,0.7}

\title{Optimal steering of a linear stochastic system\\
to a final probability distribution, Part III\thanks{
Supported in part by the
NSF under Grants ECCS-1509387,
the AFOSR under Grants FA9550-12-1-0319 and FA9550-15-1-0045, the Vincentine Hermes-Luh Chair, and by by the University of Padova Research Project CPDA 140897.}}

\author{Yongxin Chen\thanks{Y.\ Chen is with the Department of Mechanical Engineering,
University of Minnesota, Minneapolis, MN 55455; {\sf\footnotesize chen2468@umn.edu}}, Tryphon Georgiou\thanks{T.T.\ Georgiou is with the Department of Mechanical and Aerospace Engineering,
University of California, Irvine, CA 92697; {\sf\footnotesize tryphon@uci.edu}} and Michele Pavon
\thanks{M.\ Pavon is with the Dipartimento di Matematica,
Universit\`a di Padova, via Trieste 63, 35121 Padova, Italy; {\sf\footnotesize pavon@math.unipd.it}}
}
\markboth{}{}\date{}

\maketitle
\begin{abstract}
The subject of this work has its roots in the so called Schr\"odginer Bridge Problem (SBP) which asks for the most likely distribution of Brownian particles in their passage between observed empirical marginal distributions at two distinct points in time. Renewed interest in this problem was sparked by a reformulation in the language of stochastic control. In earlier works, presented as Part I and Part II, we explored a generalization of the original SBP that amounts to optimal steering of linear stochastic dynamical systems between state-distributions, at two points in time, under full state feedback. In these works the cost was quadratic in the control input. The purpose of the present work is to detail the technical steps in extending the framework to the case where a quadratic cost in the state is also present. 
In the zero-noise limit, we obtain the solution of a (deterministic) mass transport problem with general quadratic cost.
\end{abstract}

\section{Introduction}
In 1931/32, Erwin Schr\"odinger asked for the most likely evolution that a cloud of Brownian particles may have taken in between two end-point empirical marginal distributions \cite{Sch31,Sch32}. Schr\"odinger's insight was that the one-time marginal distributions along the most likely evolution can be represented as a product of two factors, a harmonic and a co-harmonic function, in close resemblance to the way the product of a quantum mechanical wave function and its adjoint produce the correct probability density. 
The 80+ year history of this so called Schr\"odinger Bridge Problem (SBP) was punctuated by advances relating SBP with large deviations theory and the Hamilton-Jacobi-Belman formalism of stochastic optimal control. More precisely, in it is original formulation, SBP seeks a probability law on path space which is closest to the prior in the sense of large deviations, i.e., closest in the relative entropy sense. Alternatively, the Girsanov transformation allows seeing this Bayesian-like estimation problem as a control problem, namely, as the problem to steer a collection of dynamical systems from an initial distribution to a final one with minimal expected quadratic input cost. The solution to the control problem generates the process and the law sought in Schr\"odinger's question.

Historically, building on the work of Jamison, Fleming, Holland, Mitter and others, Dai Pra made the connection between SBP and stochastic control \cite{Dai91}.  At about the same time, Blaquiere and others \cite{blaquiere1992controllability,DaiPav90,PavWak91,FilHonStr08}  studied the control of the Focker-Planck equation, and more recently Brockett studied the Louiville equation \cite{brockett2012notes}. The rationale for seeking to steer a stochastic or, even a deterministic system between marginal state-distributions has most eloquently been explained by Brockett, in that ``important limitations standing in the way of the wider use of optimal control [that] can be circumvented by explicitly acknowledging that in most situations the apparatus implementing the control policy will be judged on its ability to cope with a distribution of initial states, rather than a single state.''
Thus, the problem that comes into focus in this line of current research is to impose a ``soft conditioning''  in the sense that a specification for the probability distribution of the state vector is prescribed instead of initial or terminal state values. For the case of linear dynamics and quadratic input cost, the development parallels that of classical LQG regulator theory \cite{FleRis75}. More specifically, in \cite{CheGeoPav14a} the solution for quadratic input cost is provided and related to the solution of two nonlinearly-coupled homogeneous Riccati equations. The case where noise and control channels differ calls for a substantially different analysis which is given in \cite{CheGeoPav14b}. However, both \cite{CheGeoPav14a,CheGeoPav14b} do not consider penalty on state trajectories. This was discussed in \cite{HalWen16} where, rather than having a hard constraint as in the SBP on the final marginal, the authors introduce a Wasserstein distance terminal cost.  They derive necessary condition for optimality for this problem but without establishing sufficiency. Stochastic control with quadratic state-cost penalty can be given a probabilistic interpretation when the uncontrolled evolution is the law of dynamical particles/systems with creation/killing in the sense of Feynman-Kac \cite{Wak90,DaiPav90}. This was discussed in \cite{CheGeoPav14c} and necessary conditions for optimality were given there too but without establishing sufficiency. In the present work, we document fully the solution of the stochastic control problem to steer a linear system between end-point Gaussian state-distributions while minimizing a quadratic state + input cost. The solution is given in closed form by solving two matrix Riccati equations with nonlinearly coupled boundary conditions.

The paper is organized as follows. We present the problem formulation and the main results in Section \ref{sec:mainresults}. The results are used to solve the optimal mass transport problem with losses in Section \ref{sec:OMT} by taking the zero-noise limit. A numerical example is presented in Section \ref{sec:example} to highlight the results. 

\section{Main results}\label{sec:mainresults}
We consider the following optimal control problem\footnote{The choice of the time interval $[0,1]$ is without loss of generality, as the general case reduces to this by rescaling time.}
    \begin{subequations}\label{eq:problem}
    \begin{eqnarray}
    \label{eq:problem1}
    && \inf_{u\in \cU} \mE\left\{\int_0^1[\|u(t)\|^2+x(t)'Q(t)x(t)] dt\right\},
    \\&& dx(t)=A(t)x(t)dt+B(t)u(t)dt+B(t)dw(t),
    \\&& x(0)~\sim ~ \rho_0,~~~~~x(1)~\sim~ \rho_1,
    \end{eqnarray}
    \end{subequations}
where $\cU$ denotes the set of finite-energy control laws adapted to the state and $\rho_0,~\rho_1$ are zero-mean
Gaussian distributions with covariances $\Sigma_0$ and $\Sigma_1$. The optimal control for nonzero-mean cases can be obtained by introducing a suitable time-varying drift, cf.\ \cite[Remark 9]{CheGeoPav14a}. The system is assumed to be uniformly controllable in the sense that the reachability Gramian
    \[
        M(t,s)=\int_s^t \Psi(t,\tau)B(\tau)B(\tau)'\Psi(t,\tau)'d\tau
    \]
is nonsingular for all $s<t$. Here $\Psi(\cdot,\cdot)$ is the state transition matrix for $A(\cdot)$. 

Sufficient conditions for optimality were given in \cite[Proposition 1 and Section III]{CheGeoPav14c} in the form of the following two Riccati equations with coupled boundary conditions
    \begin{subequations}\label{eq:LQschrodinger}
    \begin{eqnarray}\label{eq:LQschrodinger1}
    \hspace{-0.5cm}-\dot\Pi(t)\!\!\!\!&=&\!\!\!\!A(t)'\Pi(t)\!+\!\Pi(t)A(t)\!-\!\Pi(t)B(t)B(t)'\Pi(t) \!+\!Q(t),
    \\\label{eq:LQschrodinger2}
    \hspace{-0.5cm}-\dot\HH(t)\!\!\!\!&=&\!\!\!\!A(t)'\HH(t)\!+\!\HH(t)A(t)\!+\!\HH(t)B(t)B(t)'\HH(t) \!-\!Q(t),
    \\\label{eq:LQschrodinger3}
    \Sigma_0^{-1}\!\!\!\!&=&\!\!\!\!\Pi(0)+\HH(0),
    \\\label{eq:LQschrodinger4}
    \Sigma_1^{-1}\!\!\!\!&=&\!\!\!\!\Pi(1)+\HH(1).
    \end{eqnarray}
    \end{subequations}
The special case where $Q(\cdot)\equiv 0$, i.e., the state penalty is zero, is given in  \cite{CheGeoPav14a} where a solution is given in closed form. A key contribution below is to show that the system (\ref{eq:LQschrodinger1}-\ref{eq:LQschrodinger4})
has always a solution. Thereby, under the stated conditions, an optimal control strategy always exists and turns out to be in the form of state feedback
    \begin{equation}\label{eq:optimalcontrol}
        u(t,x)=-B(t)'\Pi(t)x.
    \end{equation}
\begin{thm}\label{thm:uniquesolution}
Consider positive definite matrices $\Sigma_0, \Sigma_1$ and a pair $(A(\cdot),B(\cdot))$ that is uniformly controllable. The coupled system of Riccati equations (\ref{eq:LQschrodinger1}-\ref{eq:LQschrodinger4}) has a unique solution, which is determined by the
initial value problem consisting of (\ref{eq:LQschrodinger1}-\ref{eq:LQschrodinger2}) and
    \begin{subequations}\label{eq:initial}
\begin{eqnarray}
    \Pi(0)&=& \frac{\Sigma_0^{-1}}{2}-\Phi_{12}^{-1}\Phi_{11}-\Sigma_0^{-1/2}
    \\&&\times\left(\frac{I}{4}+\Sigma_0^{1/2}\Phi_{12}^{-1}\Sigma_1
    (\Phi_{12}^\prime)^{-1}\Sigma_0^{1/2}\right)^{1/2}\Sigma_0^{-1/2},
    \\
    \HH(0)&=& \Sigma_0^{-1}-\Pi(0),
    \end{eqnarray}
    \end{subequations}
where
    \begin{equation}\label{eq:statetransition}
        \Phi(t,s)=\left[
        \begin{matrix}
        \Phi_{11}(t,s) & \Phi_{12}(t,s)\\
        \Phi_{21}(t,s) & \Phi_{22}(t,s)
        \end{matrix}\right]
    \end{equation}
is a state transition matrix corresponding to $\partial \Phi(t,s)/\partial t = M(t)\Phi(t,s)$ with $\Phi(s,s)=I$ and
    \[
       M(t)= \left[
        \begin{matrix}A(t) & -B(t)B(t)'\\-Q(t)& -A(t)'\end{matrix}
        \right],
    \]
and where
    \[
        \left[
        \begin{matrix}
        \Phi_{11}& \Phi_{12}\\
        \Phi_{21}& \Phi_{22}
        \end{matrix}\right]
        :=\left[
        \begin{matrix}
        \Phi_{11}(1,0) & \Phi_{12}(1,0)\\
        \Phi_{21}(1,0) & \Phi_{22}(1,0)
        \end{matrix}\right].
    \]
\end{thm}

We continue with
two technical lemmas needed in the proof of the theorem.
\begin{lemma}\label{lem:matrixequation}
Given positive definite matrices $X, Y$,
    \begin{eqnarray}
    \nonumber
    &&\hspace{-0.9cm}Y^{1/2}(Y^{-1/2}X^{-1}Y^{-1/2}\!+\!\frac{1}{4}Y^{-1/2}X^{-1}Y^{-1}X^{-1}Y^{-1/2})^{1/2}Y^{1/2}
    \\&&=
    X^{-1/2}(\frac{I}{4}+X^{1/2}YX^{1/2})^{1/2}X^{-1/2}.\label{eq:matrixidentity}
    \end{eqnarray}
\end{lemma}
\begin{proof}
Multiplying both sides of \eqref{eq:matrixidentity} by $X^{1/2}$ from both left and right we obtain
    \[
        G((G'G)^{-1}+\frac{1}{4}(G'G)^{-2})^{1/2}G'=(\frac{I}{4}+GG')^{1/2},
    \]
where $G$ denotes $X^{1/2}Y^{1/2}$. As both sides are positive definite, the above is equivalent to
    \begin{eqnarray*}
        &&\hspace{-0.9cm}G((G'G)^{-1}+\frac{1}{4}(G'G)^{-2})^{1/2}G'G((G'G)^{-1}+\frac{1}{4}(G'G)^{-2})^{1/2}G'
        \\&&=\frac{I}{4}+GG',
    \end{eqnarray*}
by taking the square of both sides. Since $G'G$ commutes with $((G'G)^{-1}+\frac{1}{4}(G'G)^{-2})^{1/2}$, the LHS of the above is equal to
    \[
        GG'G((G'G)^{-1}+\frac{1}{4}(G'G)^{-2})G'=GG'+\frac{I}{4},
    \]
which completes the proof.
\end{proof}
\begin{lemma}\label{lem:transition}
The entries of the state transition matrix in \eqref{eq:statetransition} satisfy:
    \begin{subequations}\label{eq:tranmatrix}
    \begin{eqnarray}
    \label{eq:tranmatrix1}
    \Phi_{11}(t,s)'\Phi_{22}(t,s)-\Phi_{21}(t,s)'\Phi_{12}(t,s) &=& I,
    \\\label{eq:tranmatrix2}
    \Phi_{12}(t,s)'\Phi_{22}(t,s)-\Phi_{22}(t,s)'\Phi_{12}(t,s) &=& 0,
    \\\label{eq:tranmatrix3}
    \Phi_{21}(t,s)'\Phi_{11}(t,s)-\Phi_{11}(t,s)'\Phi_{21}(t,s) &=& 0,
    \\\label{eq:tranmatrix4}
    \Phi_{11}(t,s)\Phi_{22}(t,s)'-\Phi_{12}(t,s)\Phi_{21}(t,s)' &=& I,
    \\\label{eq:tranmatrix5}
    \Phi_{12}(t,s)\Phi_{11}(t,s)'-\Phi_{11}(t,s)\Phi_{12}(t,s)' &=& 0,
    \\\label{eq:tranmatrix6}
    \Phi_{21}(t,s)\Phi_{22}(t,s)'-\Phi_{22}(t,s)\Phi_{21}(t,s)' &=& 0,
    \end{eqnarray}
    \end{subequations}
for all $s\le t$.
Moreover, both $\Phi_{12}(t,s)$ and $\Phi_{11}(t,s)$ are invertible for all $s<t$, and $(\Phi_{12}(t,0)^{-1}\Phi_{11}(t,0))^{-1}$ is monotonically decreasing function in the positive definite sense with left limit $0$ as $t\searrow 0$.

\end{lemma}
\begin{proof}
A direct consequence of the fact that $M(t)J+JM(t)'=0$, with $J=\left[\begin{matrix}0&I\\-I&0\end{matrix}\right]$, is that
    \begin{eqnarray}
    \nonumber
        J_1(t,s)\!\!\!&:=&\!\!\!
        \left[
        \begin{matrix}
        \Phi_{11}(t,s)' \!&\! \Phi_{21}(t,s)'\\\Phi_{12}(t,s)' \!&\!\Phi_{22}(t,s)'
        \end{matrix}
        \right]\!\!\!
        \left[
        \begin{matrix}
        0 & I \\-I &0
        \end{matrix}
        \right]\!\!\!
        \left[
        \begin{matrix}
        \Phi_{11}(t,s) \!&\! \Phi_{12}(t,s)\\\Phi_{21}(t,s)\!&\!\Phi_{22}(t,s)
        \end{matrix}
        \right] 
        \\\nonumber\\\!\!\!&\equiv&\!\!\! J.\label{eq:J1}
    \end{eqnarray}
To see this, note that $J_1(s,s)=J$ while
    \[
        \frac{\partial}{\partial t}J_1(t,s)=0.
    \]
Likewise,
    \begin{eqnarray}
    	\nonumber
        J_2(t,s)\!\!\!&=&\!\!\!
        \left[
        \begin{matrix}
        \Phi_{11}(t,s) \!&\! \Phi_{12}(t,s)\\\Phi_{21}(t,s)\!&\!\Phi_{22}(t,s)
        \end{matrix}
        \right]\!\!\!
        \left[
        \begin{matrix}
        0 & I \\-I &0
        \end{matrix}
        \right]\!\!\!
        \left[
        \begin{matrix}
        \Phi_{11}(t,s)' \!&\! \Phi_{21}(t,s)'\\\Phi_{12}(t,s)'\!&\!\Phi_{22}(t,s)'
        \end{matrix}
        \right] 
        \\\nonumber\\\!\!\!&\equiv&\!\!\! J.\label{eq:J2}
    \end{eqnarray}
    Then, \eqref{eq:J1} gives \eqref{eq:tranmatrix1}-\eqref{eq:tranmatrix3} and \eqref{eq:J2} gives
\eqref{eq:tranmatrix4}-\eqref{eq:tranmatrix6}.

We next show both $\Phi_{12}(t,s)$ and $\Phi_{11}(t,s)$ are invertible for all $s<t$. Let
    \[
        T(t,s)=\Phi_{11}(t,s)^{-1}\Phi_{12}(t,s).
    \]
Since $\Phi_{11}(s,s)=I$, by continuity $T(t,s)$ is well-defined for $|t-s|$ sufficiently small. What's more, $T(t,s)$ is symmetric by \eqref{eq:tranmatrix5}. Taking the derivative of $T$ with respect to $s$ yields
    \begin{eqnarray*}
        \frac{\partial}{\partial s}T(t,s)
        &=&A(s)T(t,s)+T(t,s)A(s)'+B(s)B(s)'\\&&-T(t,s)Q(s)T(t,s).
    \end{eqnarray*}
This together with the initial condition $T(t,t)=0$ and the assumption that $(A, B)$ is controllable lead to
    \[
        T(t,s)<0
    \]
for all $s<t$, which implies that both $\Phi_{11}(t,s)$ and $\Phi_{12}(t,s)$ are invertible for all $s<t$.

Finally, taking the derivative of $T$ with respect to $t$ we obtain
    \begin{eqnarray*}
        \frac{\partial}{\partial t}T(t,s)
        \!\!\!&=&\!\!\! -\Phi_{11}(t,s)^{-1}\frac{\partial}{\partial t}
        \Phi_{11}(t,s)\Phi_{11}(t,s)^{-1}\Phi_{12}(t,s)
        \\&&\!\!\!+\Phi_{11}(t,s)^{-1}\frac{\partial}{\partial t}\Phi_{12}(t,s)
        \\&=&\!\!\! \Phi_{11}(t,s)^{-1}B(t)B(t)'(\Phi_{21}(t,s)
        \Phi_{11}(t,s)^{-1}\Phi_{12}(t,s)
        \\&&\!\!\!-\Phi_{22}(t,s))
        \\&=&\!\!\! \Phi_{11}(t,s)^{-1}B(t)B(t)'(\Phi_{21}(t,s)
        \Phi_{12}(t,s)'(\Phi_{11}(t,s)^{-1})'
        \\&&-\Phi_{22}(t,s))
        \\&=& -\Phi_{11}(t,s)^{-1}B(t)B(t)'(\Phi_{11}(t,s)^{-1})'\le 0,
    \end{eqnarray*}
where we used \eqref{eq:tranmatrix4} and the fact
that $\Phi_{11}(t,s)^{-1}\Phi_{12}(t,s)$ is symmetric in the last two steps. Therefore, we conclude that $T(t,s)$ is continuous monotonically decreasing function of $t (>s)$ in the positive-definite sense, with left limit $T(s,s)=0$ at $t=s$.
\end{proof}

\begin{proof}[proof of Theorem \ref{thm:uniquesolution}]
The basic idea
is to recast
the Riccati equations (\ref{eq:LQschrodinger1}-\ref{eq:LQschrodinger2}) as linear differential equations in the standard manner. 
To this end, let $[X(t)',Y(t)']'$ be the solution of
    \begin{equation}\label{eq:lineardynamicsPi}
    \left[\begin{array}{c}\dot{X}\\\dot{Y}\end{array}\right]
    =
    \left[
        \begin{matrix}A(t) & -B(t)B(t)'\\-Q(t)& -A(t)'\end{matrix}
    \right]
    \left[\begin{array}{c}X\\Y\end{array}\right].
    \end{equation}
Then
    \begin{equation}\label{eq:solutionPi}
    \Pi(t)=Y(t)X(t)^{-1}
    \end{equation}
is a solution to the Riccati equation \eqref{eq:LQschrodinger1}
provided that $X(t)$ is invertible for all $t$. To see this, differentiate \eqref{eq:solutionPi} to obtain
    \begin{eqnarray*}
        -\dot{\Pi}(t)&=& -\dot{Y}(t)X(t)^{-1}+Y(t)X(t)^{-1}\dot{X}(t)X(t)^{-1}
        \\
        &=& (QX+A'Y)X^{-1}+YX^{-1}(AX-BB'Y)X^{-1}
        \\
        &=& A'YX^{-1}+YX^{-1}A-YX^{-1}BB'YX^{-1}+Q
        \\
        &=& A'\Pi(t)+\Pi(t)A-\Pi(t)BB'\Pi(t)+Q,
    \end{eqnarray*}
which coincides with \eqref{eq:LQschrodinger1}. Similarly, let
    \begin{equation}\label{eq:solutionH}
    \HH(t)=-(\hat{X}(t)')^{-1}\hat{Y}(t)'
    \end{equation}
with
    \begin{equation}\label{eq:lineardynamicsH}
    \left[\begin{array}{c}\dot{\hat{X}}\\\dot{\hat{Y}}\end{array}\right]
    =
    \left[
        \begin{matrix}A(t) & -B(t)B(t)'\\-Q(t)& -A(t)'\end{matrix}
    \right]
    \left[\begin{array}{c}\hat{X}\\\hat{Y}\end{array}\right]
    \end{equation}
is a solution to \eqref{eq:LQschrodinger2} provided that $\hat{X}(t)$ is invertible for all $t$. Plugging \eqref{eq:solutionPi} and \eqref{eq:solutionH} into the boundary conditions \eqref{eq:LQschrodinger3} and \eqref{eq:LQschrodinger4} yields
    \begin{eqnarray*}
        \Sigma_0^{-1}&=&Y(0)X(0)^{-1}-(\hat{X}(0)')^{-1}\hat{Y}(0)',
        \\
        \Sigma_1^{-1}&=&Y(1)X(1)^{-1}-(\hat{X}(1)')^{-1}\hat{Y}(1)'.
    \end{eqnarray*}
Since $[X(t)',Y(t)']'$ has linear dynamics \eqref{eq:lineardynamicsPi}, we have
    \[
    \left[\begin{array}{c}X(1)\\Y(1)\end{array}\right]
    =
    \left[
        \begin{matrix}
        \Phi_{11}& \Phi_{12}\\
        \Phi_{21}& \Phi_{22}
        \end{matrix}\right]
    \left[\begin{array}{c}X(0)\\Y(0)\end{array}\right].
    \]
Similarly,
    \[
    \left[\begin{array}{c}\hat{X}(1)\\\hat{Y}(1)\end{array}\right]
    =
    \left[
        \begin{matrix}
        \Phi_{11}& \Phi_{12}\\
        \Phi_{21}& \Phi_{22}
        \end{matrix}\right]
    \left[\begin{array}{c}\hat{X}(0)\\\hat{Y}(0)\end{array}\right].
    \]
Moreover, without loss of generality, we can assume $X(0)=\hat{X}(0)=I$ because their initial values can be absorbed into $Y(0)$ and $\hat{Y}(0)$ without changing the values of $\Pi(0)$ and $\HH(0)$. In this case, the only unknowns $Y(0),\hat{Y}(0)$ are symmetric. Combining the above we obtain
    \begin{subequations}\label{eq:boundary}
    \begin{eqnarray}\label{eq:boundary1}
        \Sigma_0^{-1}&=&Y(0)-\hat{Y}(0),
        \\\nonumber
        \Sigma_1^{-1}&=&(\Phi_{21}+\Phi_{22}Y(0))(\Phi_{11}+\Phi_{12}Y(0))^{-1}
        \\&&-
        (\Phi_{11}'+\hat{Y}(0)'\Phi_{12}')^{-1}(\Phi_{21}'+\hat{Y}(0)\Phi_{22}').\label{eq:boundary2}
    \end{eqnarray}
    \end{subequations}
Multiplying \eqref{eq:boundary2} with $(\Phi_{11}'+\hat{Y}(0)\Phi_{12}')$ from the left and $(\Phi_{11}+\Phi_{12}Y(0))$ from the right yields
    \begin{eqnarray}
        \nonumber
        &&(\Phi_{11}'+\hat{Y}(0)\Phi_{12}')\Sigma_1^{-1}(\Phi_{11}+\Phi_{12}Y(0))
        \\\nonumber
        &=& (\Phi_{11}'+\hat{Y}(0)\Phi_{12}')(\Phi_{21}+\Phi_{22}Y(0))
        \\\nonumber&&-(\Phi_{21}'+\hat{Y}(0)\Phi_{22}')(\Phi_{11}+\Phi_{12}Y(0))
        \\\nonumber
        &=& \Phi_{11}'\Phi_{21}+\Phi_{11}'\Phi_{22}Y(0)+\hat{Y}(0)\Phi_{12}'\Phi_{21}
        +\hat{Y}(0)\Phi_{12}'\Phi_{22}Y(0)
        \\\nonumber
        &&-\Phi_{21}'\Phi_{11}-\Phi_{21}'\Phi_{12}Y(0)
        -\hat{Y}(0)\Phi_{22}'\Phi_{11}-\hat{Y}(0)\Phi_{22}'\Phi_{12}Y(0)
        \\\label{eq:YYhat} &=& Y(0)-\hat{Y}(0),
    \end{eqnarray}
where we use the three identities \eqref{eq:tranmatrix1}-\eqref{eq:tranmatrix3} in the last step. By \eqref{eq:boundary1}, $Y(0)$ and $\hat{Y}(0)$ can be parameterized by a symmetric matrix $Z$ as
    \begin{subequations}\label{eq:paraZ}
    \begin{eqnarray}
        Y(0)&=& Z+\frac{1}{2}\Sigma_0^{-1},
        \\
        \hat{Y}(0) &=& Z-\frac{1}{2}\Sigma_0^{-1}.
    \end{eqnarray}
    \end{subequations}
Plugging these into \eqref{eq:YYhat} yields
     \[
        \Sigma_0^{-1}=(\Phi_{11}'-\frac{1}{2}\Sigma_0^{-1}\Phi_{12}'+Z\Phi_{12}')\Sigma_1^{-1}
        (\Phi_{11}+\frac{1}{2}\Phi_{12}\Sigma_0^{-1}+\Phi_{12}Z).
     \]
Expanding it and exploring the symmetry we obtain a quadratic equation
    \begin{eqnarray*}
        &&\hspace{-1cm}Z\Phi_{12}'\Sigma_1^{-1}\Phi_{12}Z+Z\Phi_{12}'\Sigma_1^{-1}\Phi_{11}
        +\Phi_{11}'\Sigma_1^{-1}\Phi_{12}Z+\Phi_{11}'\Sigma_1^{-1}\Phi_{11}
        \\&=&\Sigma_0^{-1}+\frac{1}{4}\Sigma_0^{-1}\Phi_{12}'\Sigma_1^{-1}\Phi_{12}\Sigma_0^{-1}
    \end{eqnarray*}
on $Z$. By completion of square the left hand side is
    \begin{eqnarray*}
       && (Z+\Phi_{11}'(\Phi_{12}')^{-1})\Phi_{12}'\Sigma_1^{-1}\Phi_{12}(Z+\Phi_{12}^{-1}\Phi_{11}).
    \end{eqnarray*}
Note here we use the fact that $\Phi_{12}$ is invertible (see Lemma \ref{lem:transition}). By \eqref{eq:tranmatrix5}, $\Phi_{12}^{-1}\Phi_{11}$ is symmetric, therefore
    \begin{eqnarray*}
        (T^{-1/2}(Z+\Phi_{12}^{-1}\Phi_{11})T^{-1/2})^2&=&\\
       &&\hspace*{-3cm} T^{-1/2}(\Sigma_0^{-1}+\frac{1}{4}\Sigma_0^{-1}T^{-1}\Sigma_0^{-1})T^{-1/2},
    \end{eqnarray*}
where $T=(\Phi_{12}'\Sigma_1^{-1}\Phi_{12})^{-1}$. It follows that the only solutions are
    \begin{eqnarray*}
        Z_{\pm}&=&-\Phi_{12}^{-1}\Phi_{11}\pm
        T^{1/2}(T^{-1/2}\Sigma_0^{-1}T^{-1/2}+\\&&
       \hspace*{15pt} \frac{1}{4}T^{-1/2}\Sigma_0^{-1}T^{-1}\Sigma_0^{-1}T^{-1/2})^{1/2}
        T^{1/2}.
    \end{eqnarray*}
Since $\Sigma_0$ and $T$ are positive definite, we can apply Lemma \ref{lem:matrixequation} and arrive at
    \[
        Z_{\pm}=-\Phi_{12}^{-1}\Phi_{11}\pm\Sigma_0^{-1/2}(\frac{I}{4}+
        \Sigma_0^{1/2}\Phi_{12}^{-1}\Sigma_1(\Phi_{12}')^{-1}\Sigma_0^{1/2})^{1/2}\Sigma_0^{-1/2}.
    \]
The unknowns $Y(0)$ and $\hat{Y}(0)$ can be obtained by plugging the above into \eqref{eq:paraZ}.

We next show that when $Z=Z_-$, the solutions to \eqref{eq:lineardynamicsPi} and \eqref{eq:lineardynamicsH} satisfy that $X(t)$ and $\hat{X}(t)$ are invertible for all $t\in [0, 1]$, while this is not the case when $Z=Z_+$. This implies that when $Z=Z_-$, the pair $(\Pi(\cdot), \HH(\cdot))$ in \eqref{eq:solutionPi} and \eqref{eq:solutionH} is well defined and solves the coupled Riccati equations \eqref{eq:LQschrodinger},  whereas, $\Pi(\cdot)$ or $\HH(\cdot)$ would have finite escape time when $Z=Z_+$.

By \eqref{eq:lineardynamicsPi}, recalling the initial condition $X(0)=I$,
    \begin{eqnarray*}
        X(t)&=&\Phi_{11}(t,0)+\Phi_{12}(t,0)Y(0)\\&=&\Phi_{11}(t,0)+\Phi_{12}(t,0)(\frac{1}{2}
        \Sigma_0^{-1}+Z).
    \end{eqnarray*}
Since $\Phi_{12}(t,0)$ is nonsingular for all $t\in (0,1]$, it follows
    \[
        \Phi_{12}(t,0)^{-1}X(t)=\Phi_{12}(t,0)^{-1}\Phi_{11}(t,0)+\frac{1}{2}\Sigma_0^{-1}+Z.
    \]
First, when $Z=Z_-$, we have
    \begin{eqnarray*}
        \Phi_{12}(t,0)^{-1}X(t)\!\!\!\!&=&\!\!\!\!\Phi_{12}(t,0)^{-1}\Phi_{11}(t,0)-\Phi_{12}^{-1}\Phi_{11}
        +\frac{1}{2}\Sigma_0^{-1}
        \\&&\hspace{-1.2cm}-\Sigma_0^{-1/2}(\frac{I}{4}\!+\!
        \Sigma_0^{1/2}\Phi_{12}^{-1}\Sigma_1(\Phi_{12}')^{-1}\Sigma_0^{1/2})^{1/2}\Sigma_0^{-1/2}.
    \end{eqnarray*}
By Lemma \ref{lem:transition},
    \[
        \Phi_{12}(t,0)^{-1}\Phi_{11}(t,0)\le\Phi_{12}(1,0)^{-1}\Phi_{11}(1,0)=\Phi_{12}^{-1}\Phi_{11},
    \]
therefore, for any $t\in (0,1]$,
    \begin{eqnarray*}
        \Phi_{12}(t,0)^{-1}X(t)&\le&
        \frac{1}{2}\Sigma_0^{-1}-\Sigma_0^{-1/2}(\frac{I}{4}
        +\\&&\hspace*{-15pt}
        \Sigma_0^{1/2}\Phi_{12}^{-1}\Sigma_1(\Phi_{12}')^{-1}\Sigma_0^{1/2})^{1/2}\Sigma_0^{-1/2}
        <0
    \end{eqnarray*}
is invertible. This indicates $X(t)$ is for all $t\in [0, 1]$. On the other hand, when $Z=Z_+$,
    \begin{eqnarray*}
        \Phi_{12}(t,0)^{-1}X(t)\!\!\!\!&=&\!\!\!\!\Phi_{12}(t,0)^{-1}\Phi_{11}(t,0)-\Phi_{12}^{-1}\Phi_{11}
        +\frac{1}{2}\Sigma_0^{-1}+
        \\&&\hspace{-0.8cm}\Sigma_0^{-1/2}(\frac{I}{4}+
        \Sigma_0^{1/2}\Phi_{12}^{-1}\Sigma_1(\Phi_{12}')^{-1}\Sigma_0^{1/2})^{1/2}\Sigma_0^{-1/2}.
    \end{eqnarray*}
By Lemma \ref{lem:transition},
   $
        (\Phi_{12}(t,0)^{-1}\Phi_{11}(t,0))^{-1} \nearrow 0
    $
as $t\searrow 0$. Thus, for small enough $s>0$,
    $
        \Phi_{12}(s,0)^{-1}X(s)
    $
 is symmetric and negative definite.
But for $t=1$,
    \begin{eqnarray*}
        \Phi_{12}(1,0)^{-1}X(1)&=&
        \frac{1}{2}\Sigma_0^{-1}+\Sigma_0^{-1/2}(\frac{I}{4}
        +\\&&\hspace*{-15pt}
        \Sigma_0^{1/2}\Phi_{12}^{-1}\Sigma_1(\Phi_{12}')^{-1}\Sigma_0^{1/2})^{1/2}\Sigma_0^{-1/2}
        >0.
    \end{eqnarray*}
Hence, by continuity of $X(t)$ we conclude that there exists $\tau\in (s,1)$ such that $X(\tau)$ is singular. This implies that $\Pi(t)$  grows unbounded at $t=\tau$.
An analogous argument can be carried out for $\hat{X}$ and $\HH$.
 Finally, setting $Z=Z_-$ into \eqref{eq:paraZ} and recalling that $X(0)=\hat{X}(0)=I$ we obtain
    \begin{align*}
    \Pi(0)&= \frac{\Sigma_0^{-1}}{2}-\Phi_{12}^{-1}\Phi_{11}\\
    & \hspace*{10pt}-\Sigma_0^{-1/2}
    \left(\frac{I}{4}
    +\Sigma_0^{1/2}\Phi_{12}^{-1}\Sigma_1
    (\Phi_{12}^\prime)^{-1}\Sigma_0^{1/2}\right)^{1/2}\Sigma_0^{-1/2},
    \\
    \HH(0)&= \Sigma_0^{-1}-\Pi(0).
    \end{align*}
This completes the proof.
\end{proof}
The result for the $Q\equiv 0$ in \cite[Proposition 4, Remark 6]{CheGeoPav14a} can be recovered as a special case of the Theorem \ref{thm:uniquesolution}.
\begin{cor}
Given $\Sigma_0, \Sigma_1>0$ and controllable pair $(A(\cdot),B(\cdot))$, the Riccati equations \eqref{eq:LQschrodinger} with $Q\equiv 0$ has a unique solution, which is determined by the initial conidtions
    \begin{eqnarray*}
    \Pi(0)\!\!\!\!&=&\!\!\!\! \frac{\Sigma_0^{-1}}{2}+\Psi(1,0)'M(1,0)^{-1}\Psi(1,0)
    -\Sigma_0^{-1/2}\left(\frac{I}{4}\right.
    \\&&\hspace{-1cm}\left.+\Sigma_0^{1/2}\Psi(1,0)'M(1,0)^{-1}
    \Sigma_1 M(1,0)^{-1}\Psi(1,0)\Sigma_0^{1/2}\right)^{1/2}\Sigma_0^{-1/2},
    \\
    \HH(0)\!\!\!\!&=&\!\!\!\! \Sigma_0^{-1}-\Pi(0),
    \end{eqnarray*}
where $\Psi$ is the state transition matrix of $(A, B)$ and $M$ is the corresponding reachability Gramian.
\end{cor}
\begin{proof}
Simply note that when $Q\equiv 0$ we have
   $
        \Phi_{11}=\Psi(1,0),
    $
and
   $
        \Phi_{12}=-M(1,0)(\Psi(1,0)')^{-1}.
   $
\end{proof}

\section{Zero-noise limit and OMT with losses}\label{sec:OMT}
The zero-noise limit of the optimal steering problem \ref{eq:problem} is a optimal mass transport problem with general quadratic cost. That is, the solution of
    \begin{subequations}\label{eq:optimalcontroleps}
    \begin{eqnarray}
    && \inf_{u\in \cU} \mE\left\{\int_0^1[\|u(t)\|^2+x(t)'Q(t)x(t)] dt\right\},
    \\&& dx(t)=A(t)x(t)dt+B(t)u(t)dt+\sqrt{\epsilon}B(t)dw(t),
    \\&& x(0)~\sim ~ \rho_0,~~~~~x(1)~\sim~ \rho_1,
    \end{eqnarray}
    \end{subequations}
converges  \footnote{See \cite{CheGeoPav15b} for a precise statement of this convergence which involves weak convergence of path space probability measures and of their initial-final joint marginals.} to the solution of
    \begin{subequations}\label{eq:OMT}
    \begin{eqnarray}
    && \inf_{u\in \cU} \mE\left\{\int_0^1[\|u(t)\|^2+x(t)'Q(t)x(t)] dt\right\},
    \\&& dx(t)=A(t)x(t)dt+B(t)u(t)dt,
    \\&& x(0)~\sim ~ \rho_0,~~~~~x(1)~\sim~ \rho_1,
    \end{eqnarray}
    \end{subequations}
as $\epsilon\searrow 0$. The special case when $Q\equiv 0$ has been studied in \cite{CheGeoPav15b}. See for \cite{Mik04,MikThi08,Leo12,Leo13} the proof of the general cases.

By slightly modifying the results in Section \ref{sec:mainresults}, we can readily obtain the solution to \eqref{eq:optimalcontroleps}. The optimal control strategy for \eqref{eq:optimalcontroleps} is
    \[
        u(t,x)=-B(t)'\Pi_\epsilon(t)x
    \]
with $\Pi_\epsilon(\cdot)$ satisfying the same Riccati equation \eqref{eq:LQschrodinger1} with some proper initial condition $\Pi_\epsilon(0)$. The initial value is chosen in a way such that the covariance $\Sigma_\epsilon(\cdot)$, that is, the solution to
    \begin{equation}\label{eq:covarianceLyap}
        \dot{\Sigma}_\epsilon(t)=(A-BB'\Pi_\epsilon)\Sigma_\epsilon+\Sigma_\epsilon(A-BB'\Pi_\epsilon)'+\epsilon BB'
    \end{equation}
matches the two boundary values $\Sigma_0$ and $\Sigma_1$. Combining \eqref{eq:LQschrodinger1},\eqref{eq:covarianceLyap} and letting
    \[
        H_\epsilon(t)=\epsilon\Sigma_\epsilon^{-1}(t)-\Pi_\epsilon(t)
    \]
yield
    \[
        -\dot\HH_\epsilon(t)=A(t)'\HH_\epsilon(t)+\HH_\epsilon(t)A(t)+\HH_\epsilon(t)B(t)B(t)'\HH_\epsilon(t) -Q(t).
    \]
Therefore, to establish the optimal control for \eqref{eq:optimalcontroleps}, we only need to solve the coupled Riccati equations \eqref{eq:LQschrodinger1}-\eqref{eq:LQschrodinger2} with boundary conditions
    \[
         \epsilon\Sigma_0^{-1}=\Pi_\epsilon(0)+\HH_\epsilon(0),~~~~
         \epsilon\Sigma_1^{-1}=\Pi_\epsilon(1)+\HH_\epsilon(1).
    \]
This is nothing but Theorem \ref{thm:uniquesolution} with different boundary conditions. Therefore, The initial value for $\Pi_\epsilon(t)$ is
    \begin{align*}
    \Pi_\epsilon(0)&= \frac{\epsilon\Sigma_0^{-1}}{2}-\Phi_{12}^{-1}\Phi_{11}\\
    &\hspace*{-1cm}-\Sigma_0^{-1/2}
    \left(\frac{\epsilon^2 I}{4}+\Sigma_0^{1/2}\Phi_{12}^{-1}\Sigma_1
    (\Phi_{12}^\prime)^{-1}\Sigma_0^{1/2}\right)^{1/2}\Sigma_0^{-1/2}.
    \end{align*}
Letting $\epsilon \rightarrow 0$ we obtain that the solution to the optimal mass transport problem \eqref{eq:OMT} is $$u(t,x)=-B(t)'\Pi_0(t)x$$
where $\Pi_0(\cdot)$ satisfies the Riccati equation \eqref{eq:LQschrodinger1} with initial value
    \[
        \Pi_0(0)= -\Phi_{12}^{-1}\Phi_{11}-\Sigma_0^{-1/2}
    \left(\Sigma_0^{1/2}\Phi_{12}^{-1}\Sigma_1
    (\Phi_{12}^\prime)^{-1}\Sigma_0^{1/2}\right)^{1/2}\Sigma_0^{-1/2}.
    \]
Therefore, we established the following.
\begin{thm}\label{thm:omt}
The solution to Problem \eqref{eq:OMT} with zero-mean Gaussian marginals with covariances $\Sigma_0, \Sigma_1$ is
    \[
        u(t,x)=-B(t)'\Pi(t)x,
    \]
where $\Pi$ is the solution of the
Riccati equation \eqref{eq:LQschrodinger1} with
initial value
    \[
        \Pi(0)= -\Phi_{12}^{-1}\Phi_{11}-\Sigma_0^{-1/2}
    \left(\Sigma_0^{1/2}\Phi_{12}^{-1}\Sigma_1
    (\Phi_{12}^\prime)^{-1}\Sigma_0^{1/2}\right)^{1/2}\Sigma_0^{-1/2}.
    \]
\end{thm}

Evidently, we can similarly solve the slightly more general optimal mass transport problem
    \begin{subequations}\label{eq:OMTR}
    \begin{eqnarray}
    && \inf_{u\in \cU} \mE\left\{\int_0^1[u(t)'R(t)u(t)+x(t)'Q(t)x(t)] dt\right\},
    \\&& dx(t)=A(t)x(t)dt+B(t)u(t)dt,
    \\&& x(0)~\sim ~ \rho_0,~~~~~x(1)~\sim~ \rho_1,
    \end{eqnarray}
    \end{subequations}
where $R(t), 0\le t\le 1$ is positive definite,
as this reduces to \eqref{eq:OMT} by setting $\tilde u(t)=R(t)^{1/2}u(t)$ and $B_1(t)=B(t)R(t)^{-1/2}.$ 
More specifically, the solution to \eqref{eq:OMTR} with zero-mean Gaussian marginals having covariances $\Sigma_0, \Sigma_1$ is
given by
        $u(t,x)=-R(t)^{-1}B(t)'\Pi(t)x$,
where $\Pi$ is the solution of 
    \[
        -\dot\Pi(t)=A(t)'\Pi(t)+\Pi(t)A(t)-\Pi(t)B(t)R(t)^{-1}B(t)'\Pi(t) +Q(t)
    \]
with initial value
    \[
        \Pi(0)= -\Phi_{12}^{-1}\Phi_{11}-\Sigma_0^{-1/2}
    \left(\Sigma_0^{1/2}\Phi_{12}^{-1}\Sigma_1
    (\Phi_{12}^\prime)^{-1}\Sigma_0^{1/2}\right)^{1/2}\Sigma_0^{-1/2}.
    \]
Here
    \[
        \Phi(t,s)=\left[
        \begin{matrix}
        \Phi_{11}(t,s) & \Phi_{12}(t,s)\\
        \Phi_{21}(t,s) & \Phi_{22}(t,s)
        \end{matrix}\right]
    \]
is a state transition matrix corresponding to $\partial \Phi(t,s)/\partial t = M(t)\Phi(t,s)$ with $\Phi(s,s)=I$ and
    \[
       M(t)= \left[
        \begin{matrix}A(t) & -B(t)R(t)^{-1}B(t)'\\-Q(t)& -A(t)'\end{matrix}
        \right],
    \]
and, as before,
    \[
        \left[
        \begin{matrix}
        \Phi_{11}& \Phi_{12}\\
        \Phi_{21}& \Phi_{22}
        \end{matrix}\right]
        :=\left[
        \begin{matrix}
        \Phi_{11}(1,0) & \Phi_{12}(1,0)\\
        \Phi_{21}(1,0) & \Phi_{22}(1,0)
        \end{matrix}\right].
    \]

\section{Examples}\label{sec:example}
Consider inertial particles modeled by
   \begin{eqnarray*}
       dx(t) &=& v(t)dt \\
       dv(t) &=& u(t)dt+ \sqrt{\epsilon}dw(t),
   \end{eqnarray*}
where $u(t)$ is a control input (force) at our disposal, $x(t)$ represents the position, $v(t)$ velocity of particles, and $w(t)$ represents random exitation (corresponding to ``white noise'' forcing).
Our goal is to steer the spread of the particles from an initial Gaussian distribution with $\Sigma_0=2I$ at $t=0$ to the terminal marginal $\Sigma_1=1/4I$
for $t=1$ in a way such that the cost function \eqref{eq:problem1} is minimized.

Figure~\ref{fig:Eg1Phase1} displays typical sample paths $\{(x(t),v(t))\mid t\in[0,1]\}$ in phase space, as a function of time, that are attained using the optimal feedback strategy derived following \eqref{eq:optimalcontrol} and $Q=I$.
In all phase plots, the transparent blue ``tube'' represents the ``$3\,\sigma$'' tolerance interval. More specifically, the intersection ellipsoid between the tube and the slice plane $t$ is the set
    \[
        \left[\begin{matrix}x &v\end{matrix}\right]\Sigma(t)^{-1}\left[\begin{matrix}x\\v\end{matrix}\right]\le 3^2.
    \]
The feedback gains $K(t)=[k_1(t),\,k_2(t)]$ are shown in Figure \ref{fig:Eg1Controlfeedback} as a function of time.
Figure \ref{fig:Eg1Control1} shows the corresponding control action for each trajectory.

For comparison, Figure \ref{fig:Eg1Phase2} and Figure \ref{fig:Eg1Phase3} display typical sample paths under optimal control strategies when $Q=10I$ and $Q=-5I$ respectively. As expected, $\Sigma(\cdot)$ shrinks faster as we increase the state penalty $Q$ which is consistent with the reference evolution loosing probability mass at a higher rate at places where $x'Qx$ is large, while $\Sigma(\cdot)$ will expand first when $Q$ is negative since the particles have the tendency to stay away from the origin to reduce the cost.

To see the zero-noise limit behavior of the problems, we take different levels of noise intensity with the same $Q=I$. Figure \ref{fig:Eg2Phase1} and Figure \ref{fig:Eg2Phase2} depict the typical sample path for $\epsilon=10$ and $\epsilon=0.1$ respectively. As can be observed, the results converge to that of Problem \ref{eq:OMT}, which is shown in Figure \ref{fig:Eg2Phase3}.

\section{Conclusion}
The general theme of the work that was presented in Parts I, II, \cite{CheGeoPav14a,CheGeoPav14b} as well as in the present one, Part III, is the control of linear stochastic dynamical systems between specified distributions of their state vectors. This type of a problem represents a ``soft conditioning'' of terminal constraints that typically arise in LQG theory. It can also be seen as a precise variant of the rather indirect, and certainly less accurate, route to approximately regulate the distribution of the terminal state in LQG designs via a suitable choice of quadratic penalties. Although the development is reminiscent of classical LQG theory, in each case we studied, the key problem leads to an atypical two-point boundary value problem involving a pair of matrix Riccati equations nonlinearly coupled through their boundary conditions. 

The earlier works \cite{CheGeoPav14a,CheGeoPav14b} dealt with the case where a quadratic cost penalty is imposed on the input vector alone and, respectively, where stochastic excitation and control affect the system through the same or different channels. There is a substantial difference between the two that necessitated separate treatments. The present work, Part III, details the technical issues that arise when also a quadratic cost on the state vector is present. It is important to point out that herein we assume that noise and control input enter into the system via the same channel, i.e., same ``$B$'' matrix, very much as in the model taken in \cite{CheGeoPav14a}. The case where this is not so is currently open.

We note that the control problems to steer a stochastic linear system between terminal distributions, for the case where stochastic excitation and control input enter in the same manner, admit a Bayesian-like interpretation in that the law of the controlled system is the closest in the relative entropy sense to that of the uncontrolled system (``prior''); the presence of a state-penalty is related to creation/killing in the sense of Feynman-Kac \cite{Wak90,DaiPav90} of the uncontrolled evolution and was discussed in \cite{CheGeoPav14c}. Such an interpretation fails when the respective ``$B$''-matrices differ (as in the model in \cite{CheGeoPav14b}) because in this case the relative entropy between the two laws is infinite.

Another fruitful direction is the one taken in \cite{HalWen16} where a further relaxation of the terminal constraints was cast as a penalty on the Wasserstein distance between the terminal distribution and a pre-specified target distribution. The work in \cite{HalWen16} provides necessary conditions while no probabilistic/Bayesian interpretation of this formulation is available at present. Recent related contributions include \cite{vladimirov2015state}, where a discrete counterpart of SBP is being considered, and
\cite{bakolas2016optimal}, where the author brings in integral quadratic constraints into the corresponding covariance control problem at hand.

In all cases considered, a natural by-product is the theory to control linear {\em deterministic} systems, i.e., without stochastic excitation, between uncertain marginals for their state vectors. The underlying problem is again one of stochastic control by virtue of the random boundary state distributions. Most importantly, it represents a variant of optimal mass transport where the ``particles'' to be transported from an initial distribution to a final one obey non-trivial dynamics. Thus, the results in the present paper provide yet another generalization of optimal mass transport where the transportation cost derives from an  action functional with quadratic Lagrangian not satisfying the usual strict convexity assumption in the $\dot{x}$ variable (see \cite{CheGeoPav14e}).

\spacingset{1}

\bibliographystyle{IEEEtran}
\bibliography{refs}
\begin{figure}\begin{center}
\includegraphics[width=0.47\textwidth]{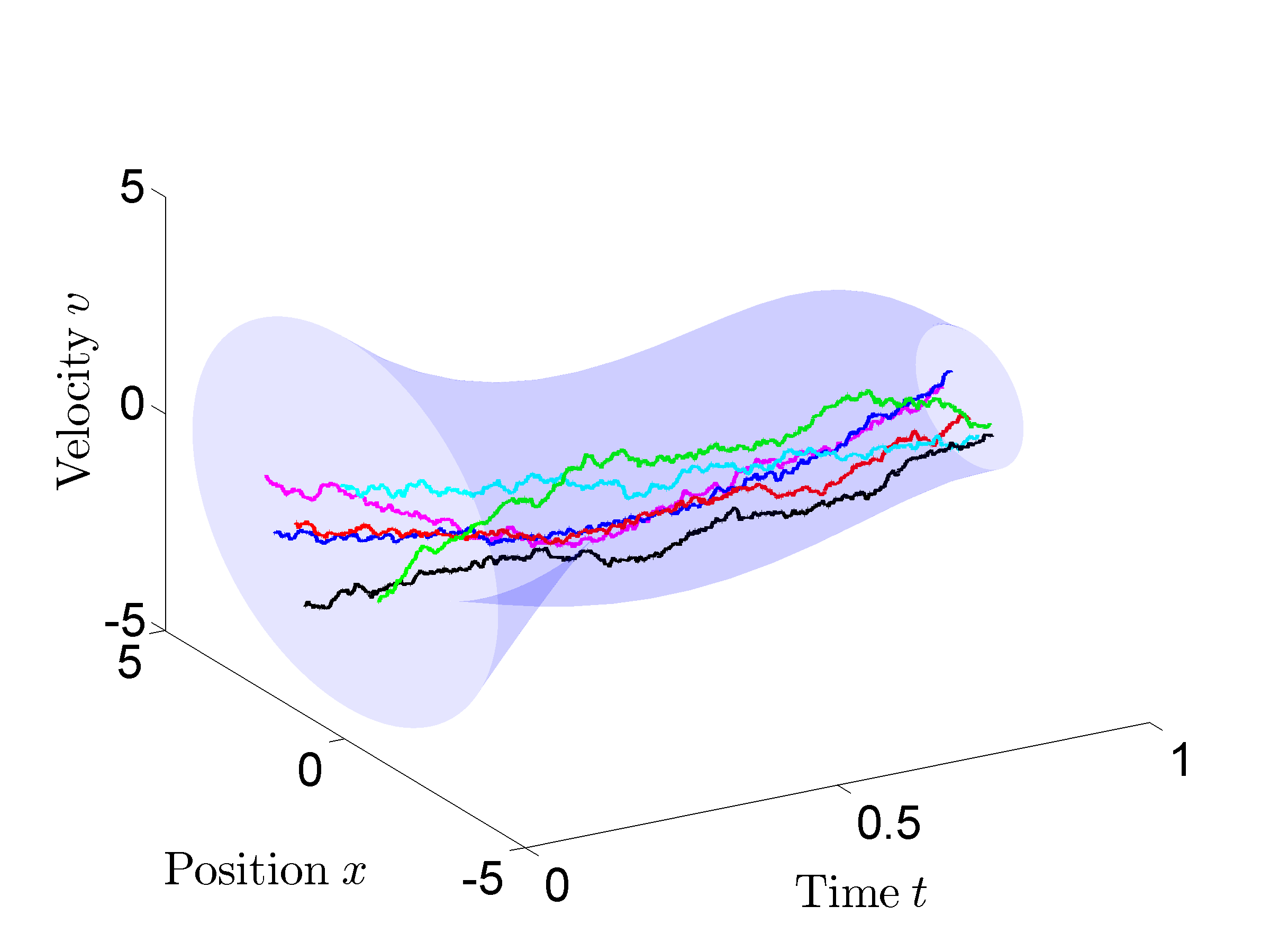}
   \caption{Inertial particles: state trajectories}
   \label{fig:Eg1Phase1}
\end{center}\end{figure}

\begin{figure}\begin{center}
\includegraphics[width=0.47\textwidth]{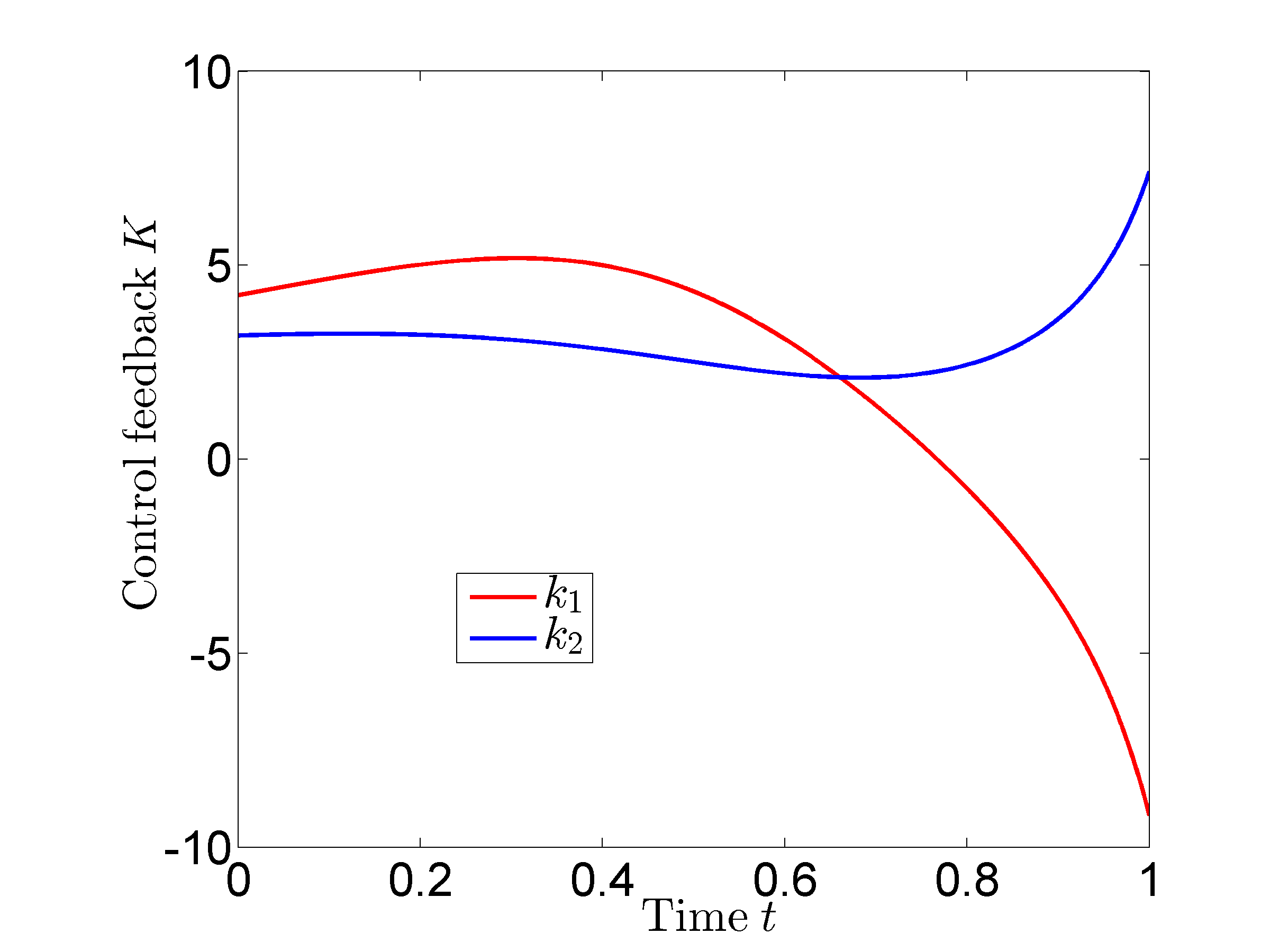}
   \caption{Inertial particles: feedback gains}
   \label{fig:Eg1Controlfeedback}
\end{center}\end{figure}

\begin{figure}\begin{center}
\includegraphics[width=0.47\textwidth]{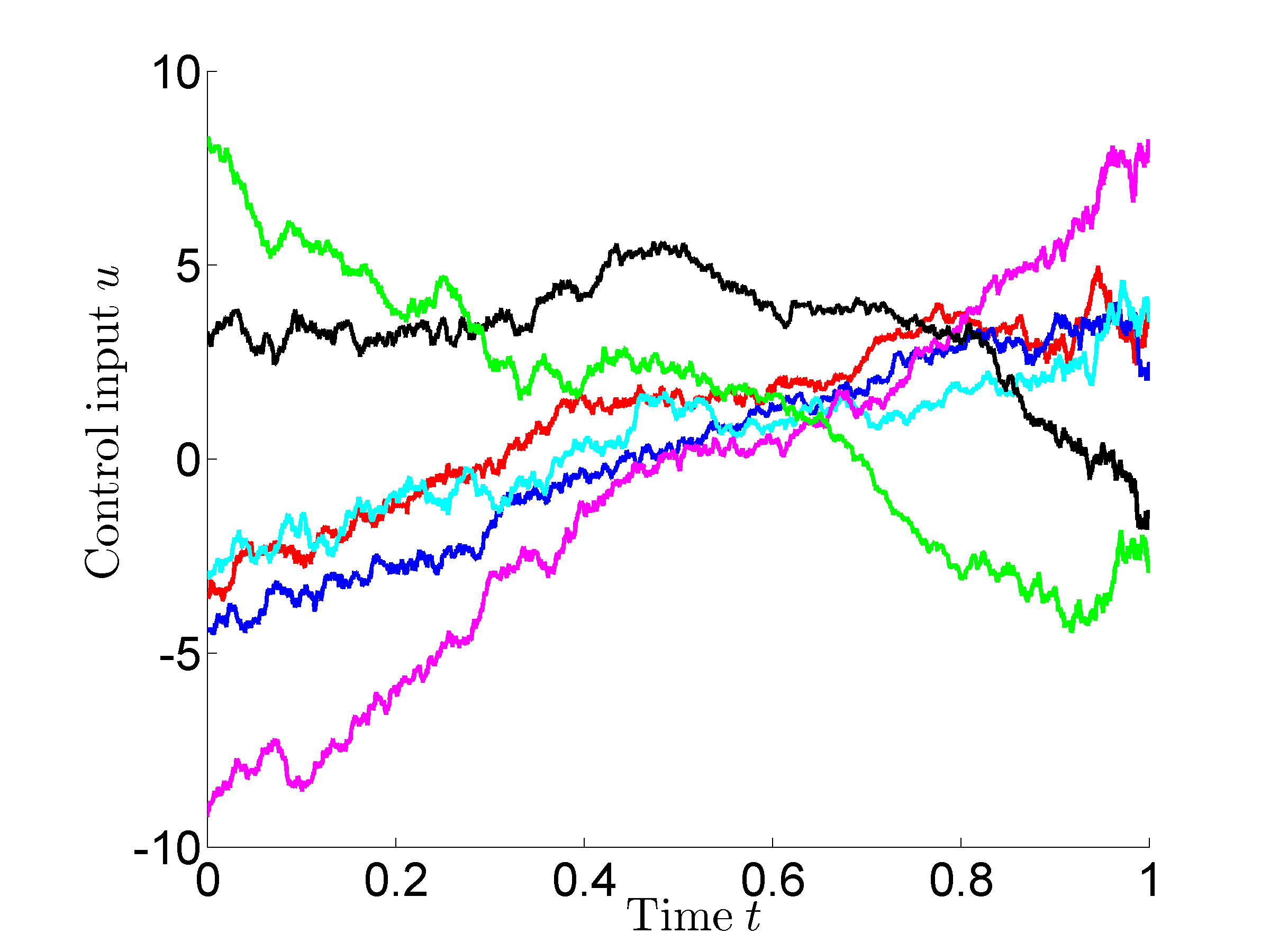}
   \caption{Inertial particles: control inputs}
   \label{fig:Eg1Control1}
\end{center}\end{figure}
\begin{figure}\begin{center}
\includegraphics[width=0.47\textwidth]{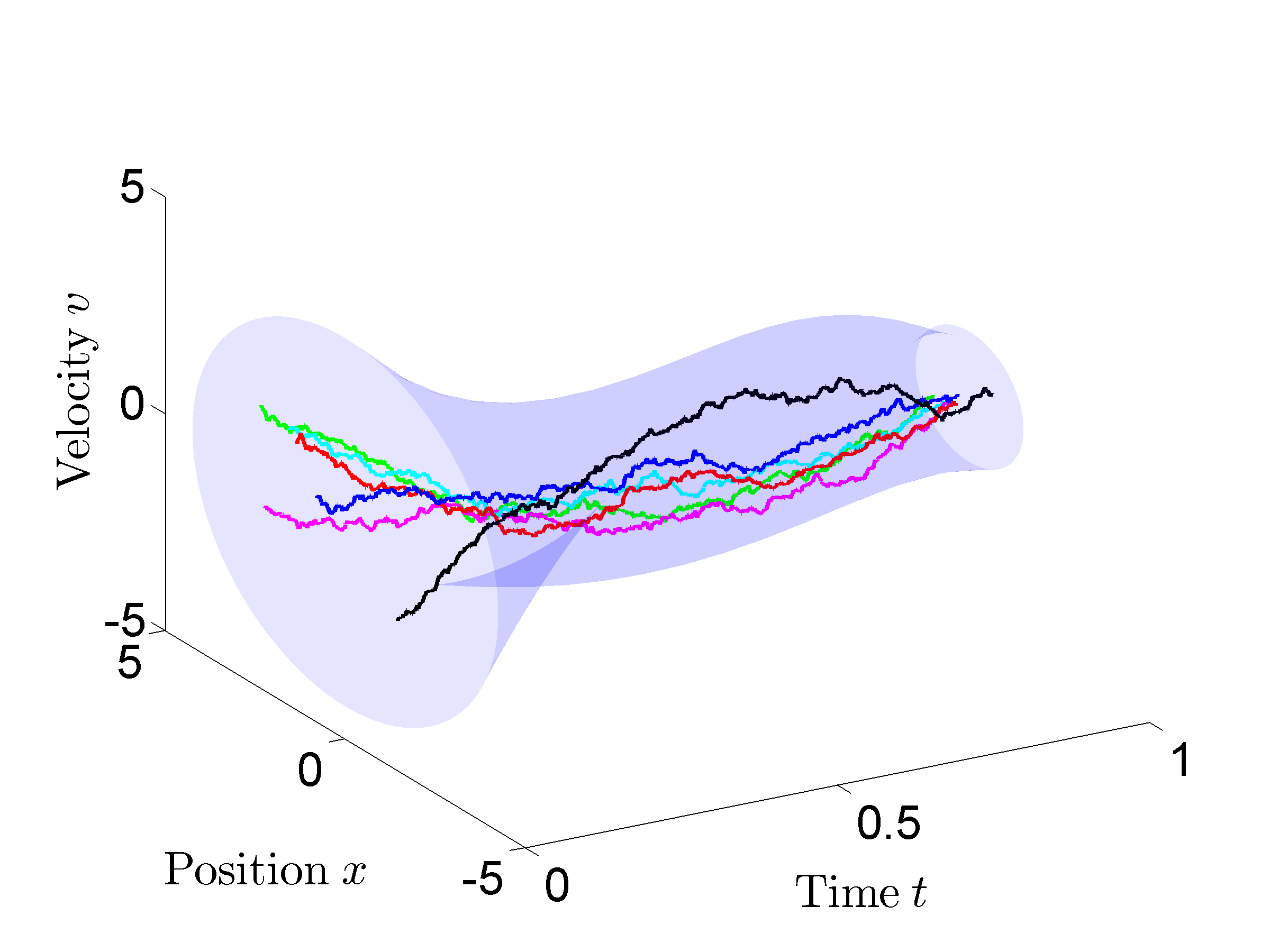}
   \caption{Inertial particles: state trajectories}
   \label{fig:Eg1Phase2}
\end{center}\end{figure}
\begin{figure}\begin{center}
\includegraphics[width=0.47\textwidth]{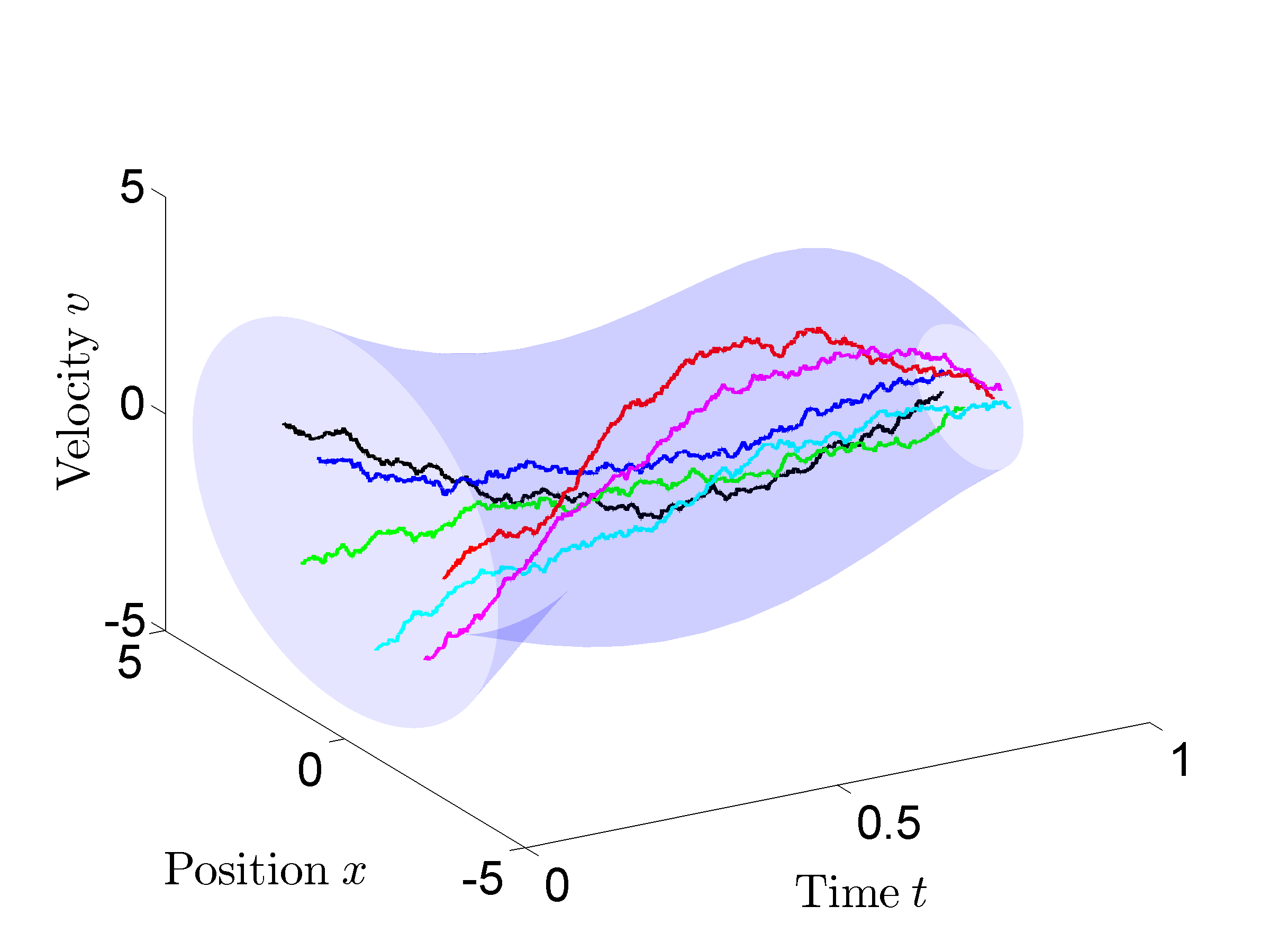}
   \caption{Inertial particles: state trajectories}
   \label{fig:Eg1Phase3}
\end{center}\end{figure}

\begin{figure}\begin{center}
\includegraphics[width=0.47\textwidth]{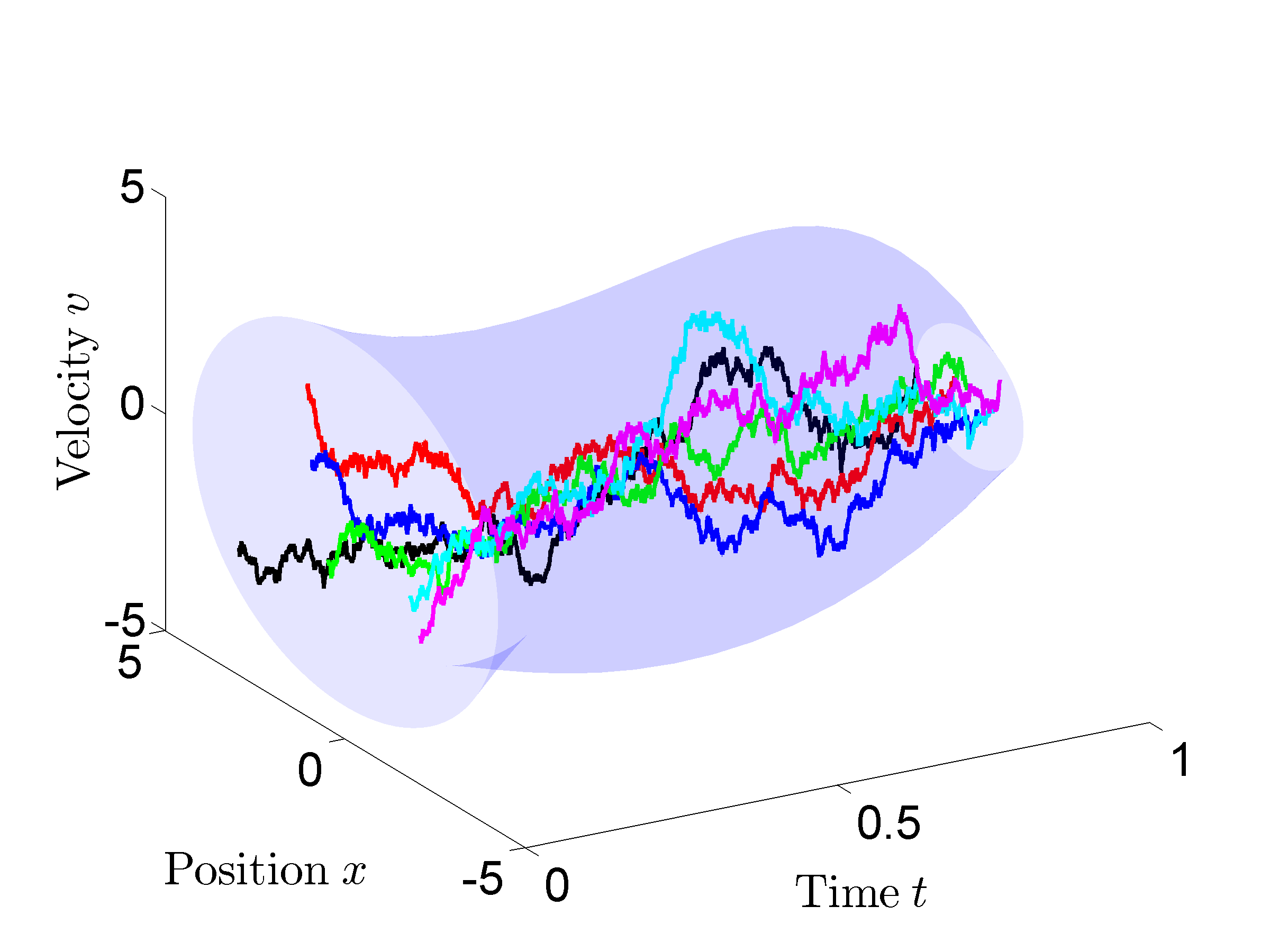}
   \caption{Inertial particles: state trajectories}
   \label{fig:Eg2Phase1}
\end{center}\end{figure}
\begin{figure}\begin{center}
\includegraphics[width=0.47\textwidth]{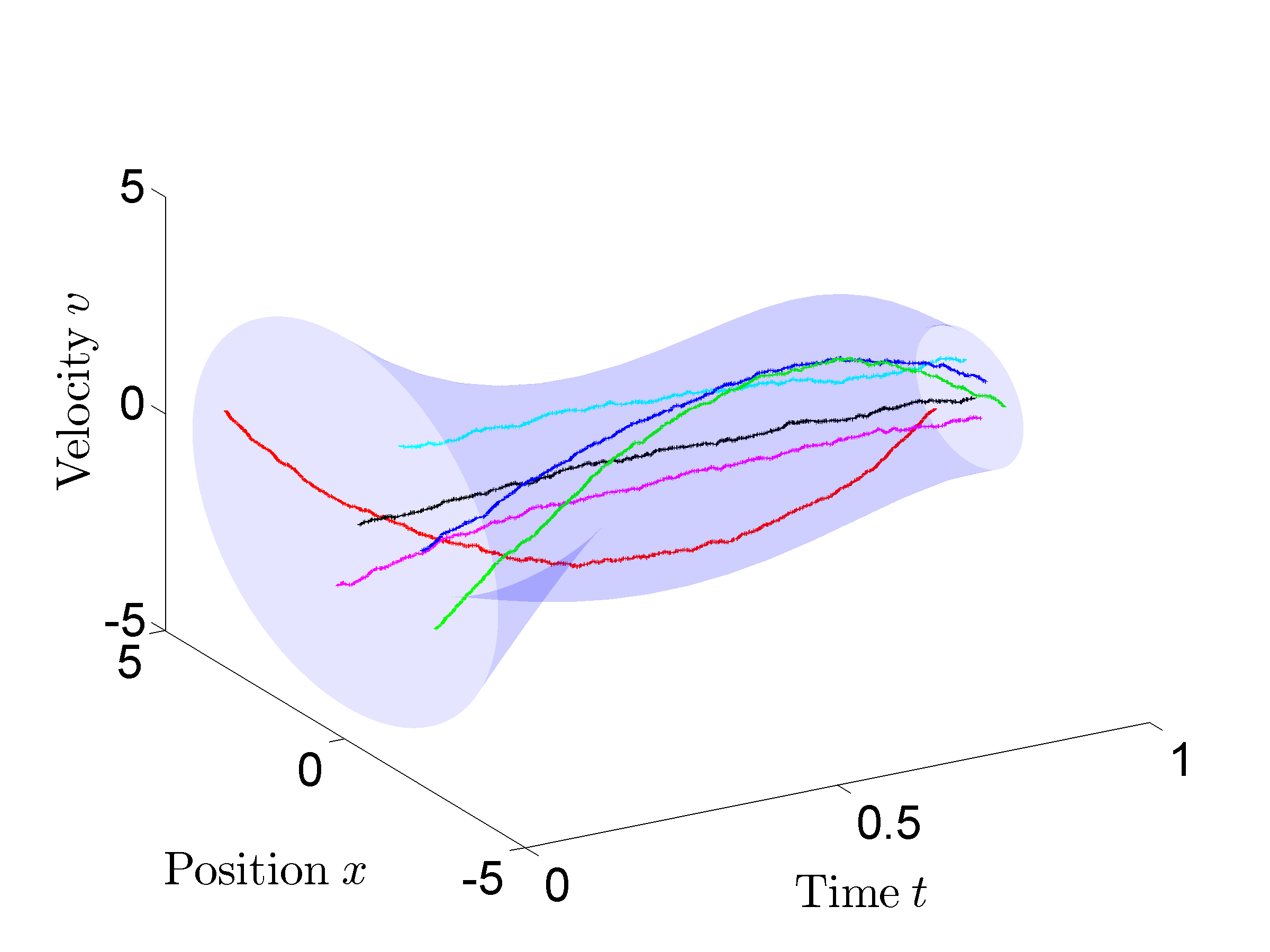}
   \caption{Inertial particles: state trajectories}
   \label{fig:Eg2Phase2}
\end{center}\end{figure}
\begin{figure}\begin{center}
\includegraphics[width=0.47\textwidth]{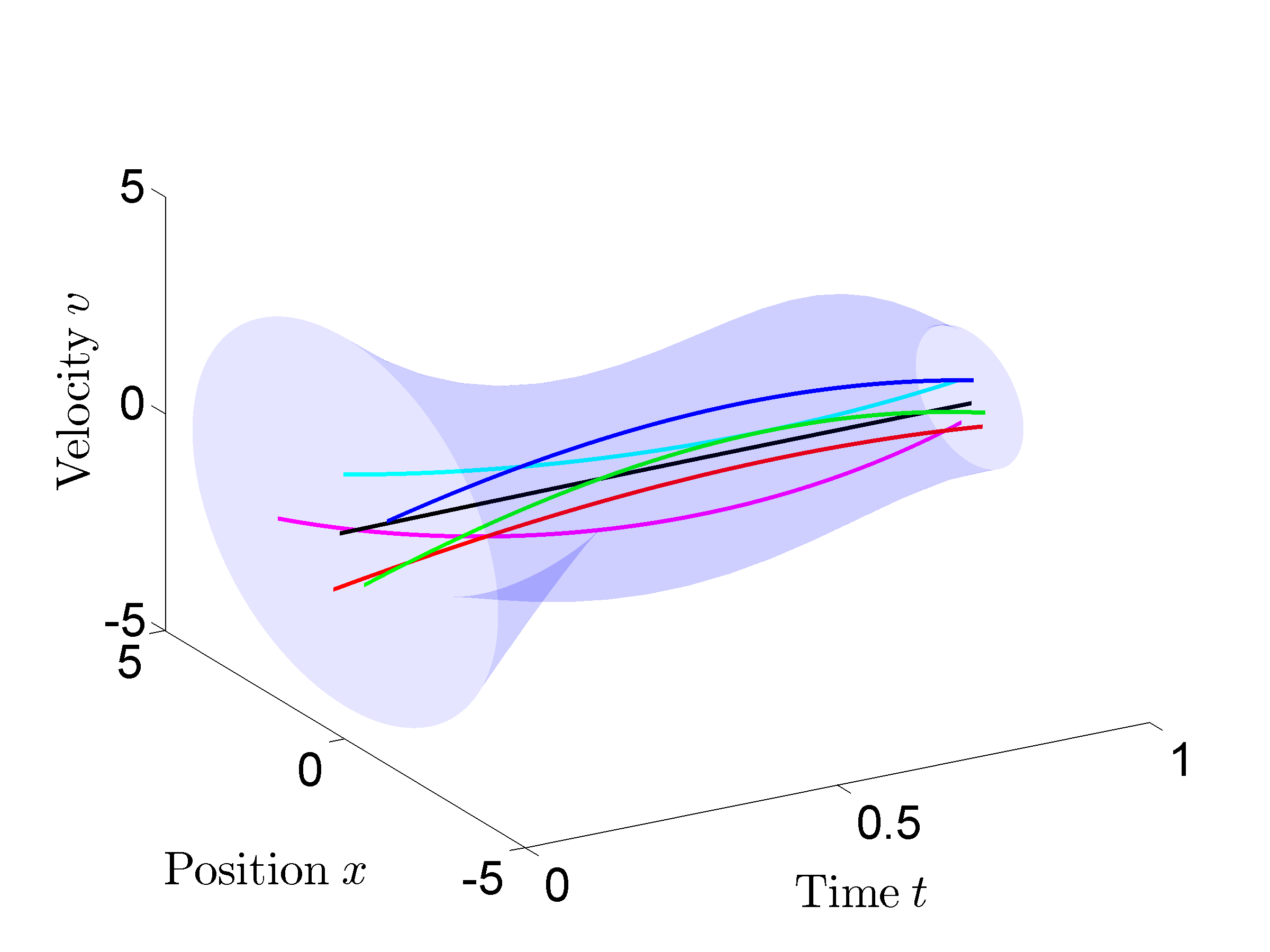}
   \caption{Inertial particles: state trajectories}
   \label{fig:Eg2Phase3}
\end{center}\end{figure}

\end{document}